\documentclass[a4paper,10pt]{article}

\usepackage[top=1in, bottom=1.05in, left=1.1in, right=1.1in]{geometry}
\usepackage{graphicx}
\usepackage{authblk}
\usepackage[round]{natbib}
\usepackage{hyperref}
\usepackage{physics}
\usepackage{lipsum}
\usepackage{amsthm}
\usepackage{booktabs} 

\usepackage{palatino}

\usepackage{amssymb,amsfonts,amsthm}
\usepackage{mathtools}
\usepackage{tikz}
\usetikzlibrary{patterns}
\usepackage{multirow}

\newtheorem{theorem}{Theorem}[section]

\newtheorem{proposition}[theorem]{Proposition}
\newtheorem{definition}[theorem]{Definition}
\newtheorem{assumption}[theorem]{Assumption}

\begin{document}

\title{On Bayesian inferential tasks with infinite-state jump processes: efficient data augmentation \thanks{Work supported by RCUK through the Horizon Digital Economy Research grants (EP/G065802/1, EP/M000877/1).} \\ {\Large Research report, June 2018}}

\author[a,b]{Iker Perez\thanks{Corresponding author address: Horizon Digital Economy Research, Triumph Road, Nottingham, NG7 2TU. Email: iker.perez@nottingham.ac.uk}}
\author[b]{Lax Chan}
\author[b,c]{Mercedes Torres Torres}
\author[d]{James Goulding}
\author[a]{Theodore Kypraios}
\affil[a]{School of Mathematical Sciences, University of Nottingham}
\affil[b]{Horizon Digital Economy Research, University of Nottingham}
\affil[c]{School of Computer Science, University of Nottingham}
\affil[d]{N-Lab, University of Nottingham}
\date{\vspace{-25pt}}

\maketitle

\begin{abstract}
Advances in sampling schemes for Markov jump processes have recently enabled multiple inferential tasks. However, in statistical and machine learning applications, we often require that these continuous-time models find support on structured and infinite state spaces. In these cases, exact sampling may only be achieved by often inefficient particle filtering procedures, and rapidly augmenting observed datasets remains a significant challenge. Here, we build on the principles of uniformization and present a tractable framework to address this problem, which greatly improves the efficiency of existing state-of-the-art methods commonly used in small finite-state systems, and further scales their use to infinite-state scenarios. We capitalize on the marginal role of variable subsets in a model hierarchy during the process jumps, and describe an algorithm that relies on measurable mappings between pairs of states and carefully designed sets of synthetic jump observations. The proposed method enables the efficient integration of slice sampling techniques and it can overcome the existing computational bottleneck. We offer evidence by means of experiments addressing inference and clustering tasks on both simulated and real data sets. 
\end{abstract}

\section{Introduction} 

Recent advances addressing conditional sampling schemes for \textit{Markov Jump Processes} (MJPs) have made inference possible in elaborate random systems with finite discrete support \citep{rao13a}. These processes often describe the dynamics that underpin many observable phenomena in diverse fields such as biology, chemistry or network evaluation (eg. \citet{hobolth2009,zhao2016bayesian,sutton2011}). However, space support in such systems is often structured and countably infinite, and inference by means of exact sampling schemes remains a significant challenge to date only addressed by particle filtering and sequential Monte Carlo methods, which are computationally intensive and find limitations due to degeneracy problems and the need for resampling \citep{hajiaghayi2014efficient,miasojedow2015particle}. It is still common practice to approximate large spaces by means of reduced sets, or to employ approximate methods relying on simplifying independence assumptions \citep{opper2008variational} or models with continuous support \citep{golightly2015bayesian}.

Here, we address the complexities posed by the data augmentation task through exact sampling, and present a tractable MCMC framework for inference with infinite-state latent structured MJPs. We target multi-component \textit{coupled} systems, whose joint behaviour often exhibits strong temporal dependencies; such as networks of queues, Markov modulated models or phase-type processes. We build on the \textit{uniformization} principles explored in \citet{zhang2017collapsed,zhang2017efficient,pan2016markov} and references therein, in combination with \textit{forward filtering backward sampling} procedures that are reported to offer better effective samples within small-scale finite-state models (cf. \citet{miasojedow2015particle}). To improve on existing procedures, we exploit the marginal role of isolated system components in the state transitions across large or infinite spaces. For this purpose, we discuss the design of mappings between a synthetic set of \textit{jump observations} and pairs of system states, and we enable the integration of \textit{slice sampling} techniques previously targeted at mixture models \citep{Walker2007,Kalli2011}. Hence, we construct an exact  sampler (in the Monte Carlo sense) that, without resorting to particle filtering procedures, can still bound the computational complexity and iteratively explore an infinite space of MJP paths by means of restricted, alternating and sequentially correlated slices. 

Finally, we conduct experiments addressing various inference and clustering tasks for service diagnosis and system strain evaluation. We provide evidence of the sampler (i) overcoming the computational bottleneck with jump models supported in large or infinite spaces and (ii) attaining significant gains in speed versus the baseline algorithm in \citet{rao13a}, for equivalent effective samples in finite or smaller systems. For this, we employ simulated and real data sets that contain (i) network performance metrics and (ii) discharge records within hospital emergency units, and we inspect the relation between processing times, effective samples and coupling amongst output traces.

\section{Markov jump processes}

An MJP is a right-continuous stochastic process $X=(X_t)_{t\geq 0}$, such that time-indexed variables $X_t$ are defined within a measurable space $(\mathcal{S},\Sigma_\mathcal{S})$. Here, $\mathcal{S}$ is a countably infinite set of possible states, and $\Sigma_\mathcal{S}$ stands for its power set. We assume the process $X$ to be time-homogeneous and governed by a \textit{generator matrix} $Q$, so that
$$\mathbb{P}(X_{t+\mathrm{d}t}=x'|X_{t}=x) = \mathbb{I}_{(x=x')} + Q_{x,x'}\mathrm{d}t + o(\mathrm{d}t)$$
for all $x,x'\in\mathcal{S}$ and $t\geq 0$; with $\mathbb{I}_{(\cdot)}$ defining a logical indicator function. The values of $Q$ describe rates for transitions within states in $X$. Also, $Q_{x,x'}\geq 0$ for all $x \neq x'$ and $Q_x \coloneqq Q_{x,x} = - \sum_{x'\in\mathcal{S}: x\neq x'} Q_{x,x'}$ so that rows sum to $0$. The time to departure or \textit{jump} from a state $x$ is exponentially distributed and its rate is given by $|Q_x|$, for all $x\in\mathcal{S}$. 
\begin{assumption} \label{AssQ}
A generator matrix $Q$ is such that its underlying process $X$ may only reach a finite countable subset of $\mathcal{S}$ within a fixed number of jumps $n\in\mathbb{N}$.
\end{assumption}
The above assumption implies that $Q$ is \textit{sparse}; however, an infinite subset of $S$ is reachable during any time interval, as there could exist infinitely many jumps within. Now, note that an MJP  is piecewise-constant, and can thus be characterized by a sequence $\boldsymbol{t}=\{t_0,\dots,t_{n}\}$  of transition times along with states $\boldsymbol{x}=\{x_0,\dots,x_{n}\}$, so that $X\equiv (\boldsymbol{t},\boldsymbol{x})$. For simplicity, we assume that the initial state $x_0\in\mathcal{S}$ is known, however it is also possible to define an initial distribution over states in $\mathcal{S}$.  The likelihood of a path $(\boldsymbol{t},\boldsymbol{x})$ over a finite time interval $[0,T]$ is such that
\begin{align}
f_{X}(\boldsymbol{t},\boldsymbol{x})  & \propto  e^{Q_{x_n}(T-t_n)} \prod_{i=1}^n Q_{x_{i-1},x_{i}} e^{Q_{x_{i-1}}(t_i-t_{i-1})}, \label{pathProbs}
\end{align}
with respect to a suitably defined base measure. Note that $n$ is the most recent jump in $X$ before time $T$.

\subsection{Observations}

Monitoring a jump process will often result in a sequence of measurements $Y=\{Y_r\}_{r\geq 1}$ at some arbitrary time points, such as sensor data in robotics or frequency recordings in audio processing tasks. These may be produced deterministically and be supported within $\mathcal{S}$, but are most often categorical, discrete or real-valued random variables governed by some conditional distribution $F_Y(y|x)= \mathbb{P}(Y_r\leq y|X_{t_r}=x)$, for observation times $t_r\geq 0$, $r\geq 1$. For instance, in Figure \ref{obsY} we observe a \textit{birth-death} process during a short time interval; there, we find dotted in red color some random observations collected at equally spaced times.

\begin{figure}[h!]
\vskip 0in
\begin{center}
	\resizebox{0.9\linewidth}{4cm}{
    \includegraphics[width=\textwidth]{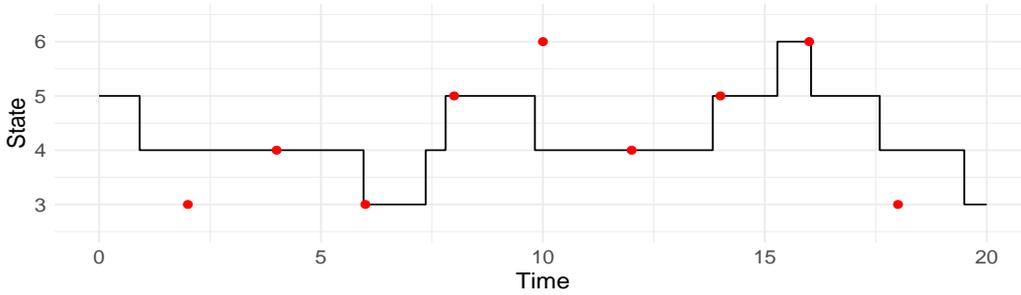} 
    }
    \vskip 0in
 \caption{Sample birth-death process realization with random observations dotted in red colour.} \label{obsY}
\end{center}
\vskip -0.2in
\end{figure}

In other instances, we may retrieve data at process jumps with some probability $q_Z\in [0,1]$, such as mutation events in genetics. We denote this by $Z=\{Z_d\}_{d\geq 1}$ and let $\mathcal{J}$ be a set of \textit{jump observations}. Often, $Z$ can be defined as $Z_d=\textstyle\mathcal{T}(\lim_{t\nearrow t_d} X_{t},X_{t_d})$ for some measurable function $\mathcal{T}:\mathcal{S}^2\rightarrow\mathcal{J}$, which determines what is observable. Also, $\mathcal{J}$ must contain a \textit{no-observation} element, denoted $\varnothing$, such that $\mathcal{T}(x,x')=\varnothing$ whenever a transition $x\rightarrow x'$ is not attainable or its underlying observation is undefined. Specifically, $\mathcal{T}(x,x)=\varnothing$ for all $x\in\mathcal{S}$. In the birth-death process shown in Figure \ref{obsY}, retrieved jump observations could, for instance, represent the \textit{sign} of jumps, so that $\mathcal{J}=\{1,-1,\varnothing\}$ and 
\begin{align}\mathcal{T}(x,x')= \begin{cases}
    x'-x  & \text{if} \quad |x'-x|=1,\\
    \varnothing & \text{otherwise.}
  \end{cases} \label{TSigns} 
  \end{align}
Assuming $q_Z=1$, all observations are retrieved and $Z=\{-1,-1,1,1,-1,1,1,-1,-1,-1\}$.

Finally, note that in large structured models, both $Y$ and $Z$ may relate to all variables in the entire model hierarchy of $X$. However, they are usually only concerned with marginal model subcomponents, such as the monitoring of an individual service node in a network of queues, or an specific sub-population in a \textit{predator-prey} model.

\subsection{Problem statement} \label{statement}

Let $O=(Y,Z)$ denote some retrieved observations from a process realization with an undetermined generator matrix $Q$, over a fixed time interval $[0,T]$. The basis for inference on the generator rates is the density $f_{Y,Z}(O|Q)$; however, this is proportional to an infinite weighted product of MJP path densities $X$ in \eqref{pathProbs}, and is thus intractable. In this paper, we address the inferential task by data augmentation, describing an efficient scheme for jump processes with infinite state support. The approach is also relevant when a set $S$ is finite but large enough to pose computational impediments.

In our experiments, we will further explore clustering exercises by means of memberships variables $\boldsymbol{c}=\{c_k\}_{k=1,\dots,K}$ for the underlying latent jump processes. Given observations $\boldsymbol{O} = \{O^k\}_{k=1\dots,K}$ produced by $L$ undetermined generator matrices $\boldsymbol{Q}=\{Q^l\}_{l=1,\dots,L}$, we have
\begin{align*}
\mathbb{P}(c_k=c|\boldsymbol{O},\boldsymbol{Q})   \propto f_{Y,Z}(O^k|Q^{c}) \cdot \pi_{c_k}(c)
\end{align*}
for $k=1,\dots,K$, with $c\in\{1,\dots,L\}$, and
\begin{align*}
f_{Q^l}(Q|\boldsymbol{O},\boldsymbol{c}) \propto \prod_{k : c_k=l} f_{Y,Z}(O^k|Q)  \cdot \pi_{Q^l}(Q)
\end{align*}
for $l=1,\dots,L$, where $\pi_{c_k}$ and $\pi_{Q^l}$ specify priors over the membership classes and generators. Here, a prior over $Q$ will factor across the individual generator rates. 

\section{Auxiliary-variable data augmentation} \label{SamplerSection}
Let $(\hat{\boldsymbol{t}},\hat{\boldsymbol{x}})$ define a renewal process over a finite time interval $[0,T]$, with $\hat{x}_i\in\mathcal{S}$ for $i\geq 0$ and such that 
\begin{itemize}
\item holding times are exponentially distributed with a fixed rate $\Omega \geq \max_{x} |Q_x|$, and
\item states form a realization from a discrete-time Markov chain, with initial state $x_0\in\mathcal{S}$ and transition probability matrix $P=I+Q/\Omega$.
\end{itemize}
\begin{proposition}
The process $(\hat{\boldsymbol{t}},\hat{\boldsymbol{x}})$ describes an \textit{augmented} MJP, and it is equivalent to $X=(\boldsymbol{t},\boldsymbol{x})$ with generator $Q$ and density function \eqref{pathProbs}.
\end{proposition}
This is a well known result and a proof of equivalence can be found in e.g. \citet{hobolth2009}. The procedure that constructs these augmented sets of times $\hat{\boldsymbol{t}}=\{\hat{t}_0,\dots,\hat{t}_{m}\}$ and states $\hat{\boldsymbol{x}}=\{\hat{x}_0,\dots,\hat{x}_{m}\}$ is commonly referred to as \textit{uniformization} (c.f. \citet{jensen1953markoff}). A \textit{uniformized} MJP path $X$ will often include self-transitions; and we refer to a transition $i$ as a \textit{virtual jump} whenever $\hat{x}_i=\hat{x}_{i-1}$. For instance, in Figure \ref{virtualJumpExpl} (left) we observe an augmented uniformized MJP trajectory for a process with $3$ states, there, we find virtual jumps at times $\hat{t}_2,\hat{t}_4$ and $\hat{t}_5$, represented by white circles on the time axis. On the right hand side plot, we observe the equivalent path after virtual times and states have been removed.
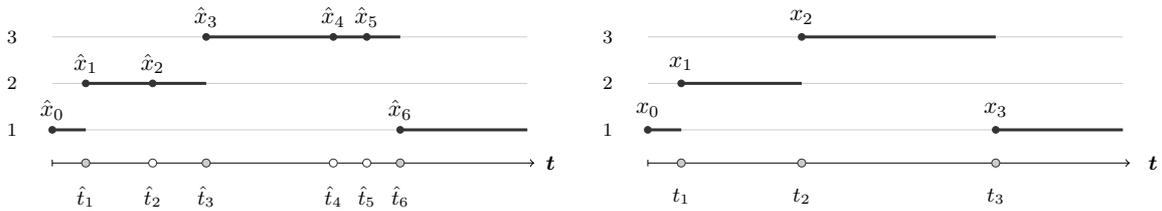
\begin{figure*}[ht]
\vskip 0.1in
\begin{center}
\resizebox{\textwidth}{!}{%
\begin{tikzpicture}

\draw[->,line width=0.10mm] (0.1cm,-0.9cm) -- ++ (7.2cm,0cm);
\draw (0.1,-0.95cm) -- ++(0cm,0.1cm);
\node at (7.55cm,-0.9cm) {\small $\boldsymbol{t}$};
\node at (-0.5cm,-0.4cm) {\footnotesize $1$};
\node at (-0.5cm,0.3cm) {\footnotesize $2$};
\node at (-0.5cm,1cm) {\footnotesize $3$};
\draw[-,draw=black!20!white,line width=0.10mm] (0.1cm,-0.4cm) -- (7.2cm,-0.4cm);
\draw[-,draw=black!20!white,line width=0.10mm] (0.1cm,0.3cm) -- (7.2cm,0.3cm); 
\draw[-,draw=black!20!white,line width=0.10mm] (0.1cm,1cm) -- (7.2cm,1cm);
\draw[-,draw=black!80!white,line width=0.45mm] (0.1cm,-0.4cm) -- (0.6cm,-0.4cm); \draw[-,draw=black!80!white,line width=0.45mm] (5.3cm,-0.4cm) -- (7.2cm,-0.4cm);
\draw[-,draw=black!80!white,line width=0.45mm] (0.6cm,0.3cm) -- (2.4cm,0.3cm); 
\draw[-,draw=black!80!white,line width=0.45mm] (2.4cm,1cm) -- (5.3cm,1cm); 
\filldraw[black!80!white] (0.1cm,-0.4cm) circle [radius=0.05cm]; 
\filldraw[black!80!white] (0.6cm,0.3cm) circle [radius=0.05cm]; 
\filldraw[black!80!white] (1.6cm,0.3cm) circle [radius=0.05cm]; 
\filldraw[black!80!white] (2.4cm,1cm) circle [radius=0.05cm]; 
\filldraw[black!80!white] (4.3cm,1cm) circle [radius=0.05cm]; 
\filldraw[black!80!white] (4.8cm,1cm) circle [radius=0.05cm]; 
\filldraw[black!80!white] (5.3cm,-0.4cm) circle [radius=0.05cm]; 
\node at (0.1cm,-0.1cm) {$\hat{x}_0$};
\node at (0.6cm,0.6cm) {$\hat{x}_1$};
\node at (1.6cm,0.6cm) {$\hat{x}_2$};
\node at (2.4cm,1.3cm) {$\hat{x}_3$};
\node at (4.3cm,1.3cm) {$\hat{x}_4$};
\node at (4.8cm,1.3cm) {$\hat{x}_5$};
\node at (5.3cm,-0.1cm) {$\hat{x}_6$};
\filldraw[fill=black!20!white,draw=black!80!white] (0.6cm,-0.9cm) circle [radius=0.06cm]; 
\filldraw[fill=black!00!white,draw=black!80!white] (1.6cm,-0.9cm) circle [radius=0.06cm]; 
\filldraw[fill=black!20!white,draw=black!80!white] (2.4cm,-0.9cm) circle [radius=0.06cm]; 
\filldraw[fill=black!00!white,draw=black!80!white] (4.3cm,-0.9cm) circle [radius=0.06cm]; 
\filldraw[fill=black!00!white,draw=black!80!white] (4.8cm,-0.9cm) circle [radius=0.06cm]; 
\filldraw[fill=black!20!white,draw=black!80!white] (5.3cm,-0.9cm) circle [radius=0.06cm]; 
\node at (0.6cm,-1.4cm) {\small $\hat{t}_1$};
\node at (1.6cm,-1.4cm) {\small $\hat{t}_2$};
\node at (2.4cm,-1.4cm) {\small $\hat{t}_3$};
\node at (4.3cm,-1.4cm) {\small $\hat{t}_4$};
\node at (4.8cm,-1.4cm) {\small $\hat{t}_5$};
\node at (5.3cm,-1.4cm) {\small $\hat{t}_6$};

\draw[->,line width=0.10mm] (9cm,-0.9cm) -- ++ (7.2cm,0cm);
\draw (9,-0.95cm) -- ++(0cm,0.1cm);
\node at (16.55cm,-0.9cm) {\small $\boldsymbol{t}$};
\node at (8.4cm,-0.4cm) {\footnotesize $1$};
\node at (8.4cm,0.3cm) {\footnotesize $2$};
\node at (8.4cm,1cm) {\footnotesize $3$};
\draw[-,draw=black!20!white,line width=0.10mm] (9cm,-0.4cm) -- ++ (7.1cm,0cm);
\draw[-,draw=black!20!white,line width=0.10mm] (9cm,0.3cm) -- ++ (7.1cm,0cm); 
\draw[-,draw=black!20!white,line width=0.10mm] (9cm,1cm) -- ++ (7.1cm,0cm);
\draw[-,draw=black!80!white,line width=0.45mm] (9cm,-0.4cm) -- (9.5cm,-0.4cm); \draw[-,draw=black!80!white,line width=0.45mm] (14.2cm,-0.4cm) -- (16.1cm,-0.4cm);
\draw[-,draw=black!80!white,line width=0.45mm] (9.5cm,0.3cm) -- (11.3cm,0.3cm); 
\draw[-,draw=black!80!white,line width=0.45mm] (11.3cm,1cm) -- (14.2cm,1cm); 
\filldraw[black!80!white] (9cm,-0.4cm) circle [radius=0.05cm]; 
\filldraw[black!80!white] (9.5cm,0.3cm) circle [radius=0.05cm]; 
\filldraw[black!80!white] (11.3cm,1cm) circle [radius=0.05cm]; 
\filldraw[black!80!white] (14.2cm,-0.4cm) circle [radius=0.05cm]; 
\node at (9cm,-0.1cm) {$x_0$};
\node at (9.5cm,0.6cm) {$x_1$};
\node at (11.3cm,1.3cm) {$x_2$};
\node at (14.2cm,-0.1cm) {$x_3$};
\filldraw[fill=black!20!white,draw=black!80!white] (9.5cm,-0.9cm) circle [radius=0.06cm]; 
\filldraw[fill=black!20!white,draw=black!80!white] (11.3cm,-0.9cm) circle [radius=0.06cm]; 
\filldraw[fill=black!20!white,draw=black!80!white] (14.2cm,-0.9cm) circle [radius=0.06cm]; 
\node at (9.5cm,-1.4cm) {\small $t_1$};
\node at (11.3cm,-1.4cm) {\small $t_2$};
\node at (14.2cm,-1.4cm) {\small $t_3$};

\end{tikzpicture}
}
\vskip 0in
\caption{On the left, augmented MJP trajectory for a process with $3$ states. On the right, equivalent path after virtual times and states are removed.} 
\label{virtualJumpExpl}
\end{center}
\vskip -0.1in
\end{figure*}

Now, let $\boldsymbol{u}$ define an auxiliary family of $m-1$ random variables $u_i\in\Sigma_{\mathcal{J}}$ such that for all  $i=1,\dots,m$, 
\begin{align}
\mathbb{P}(u_i=\{\mathcal{T}(\hat{x}_{i-1},\hat{x}_{i})\}|&\hat{x}_{i-1},\hat{x}_{i})= p / (1-q_Z)^{\mathbb{I}(\mathcal{T}(\hat{x}_{i-1},\hat{x}_i)\neq\varnothing)}, \label{aux1}
\end{align}
and
\begin{align}
\mathbb{P}(u_i=\mathcal{J}|&\hat{x}_{i-1},\hat{x}_{i})= 1 - \mathbb{P}(u_i=\{\mathcal{T}(\hat{x}_{i-1},\hat{x}_{i})\}|\hat{x}_{i-1},\hat{x}_{i}), \label{aux2}
\end{align}
with some arbitrarily fixed $p\in[0,1-q_Z)$. Whenever $u_i\neq \mathcal{J}$, this constitutes a \textit{clamped node} designed to complement real data. It is defined in order simplify the forthcoming forward filtering and backward sampling procedures in the data augmentation task. On a basic level, 
\begin{itemize}
\item an auxiliary variable $u_i=\{\mathcal{T}(\hat{x}_{i-1},\hat{x}_{i})\}$ will map a pair of states $\hat{x}_{i-1},\hat{x}_{i}\in\mathcal{S}$ to an element of $\mathcal{J}$ that holds \textit{limited} information regarding the jump across the states, and
\item an auxiliary variable $u_i=\mathcal{J}$ will hold no information. 
\end{itemize}
Analogue definitions of such variables may be found in \cite{Kalli2011,Perez2017}. For example, in the augmented process pictured in Figure \ref{virtualJumpExpl} (left diagram), by letting $\mathcal{J}=\{1,-1,\varnothing\}$ and $\mathcal{T}$ as in \eqref{TSigns}, we retrieve a random vector $\boldsymbol{u}=\{\{1\},\{\varnothing\},\{1\},\mathcal{J},\mathcal{J},\{\varnothing\}\}$, which corresponds to 
\begin{itemize}
\item $2$ jumps of magnitude $1$ and a positive sign, at times $\hat{t}_1$ and $\hat{t}_3$,
\item $2$ transitions, marked as elements $\{\varnothing\}$, that correspond to either (i) virtual jumps or (ii) jumps with a magnitude greater than $1$ (in either direction), at times $\hat{t}_2$ and $\hat{t}_6$,
\item $2$ uninformative variables that contain no jump information.
\end{itemize}
It is important to observe that, by solely looking at $\boldsymbol{u}$, we cannot retrieve the original augmented MJP path, i.e. there exist multiple \textit{compatible} sequences $\hat{\boldsymbol{x}}$ that could produce the same vector. Also, we note that by defining auxiliary variables by means of $\mathcal{T}$, we have assumed that they resemble real observations in $Z$; however, in practice, $\boldsymbol{u}$ can be tailored to each problem (see also Example \ref{experiment2}). Finally, in Figure \ref{plateNotation} we find membership, generator, auxiliary and observation variables in plate notation, for a clustering task as discussed in Subsection \ref{statement}.

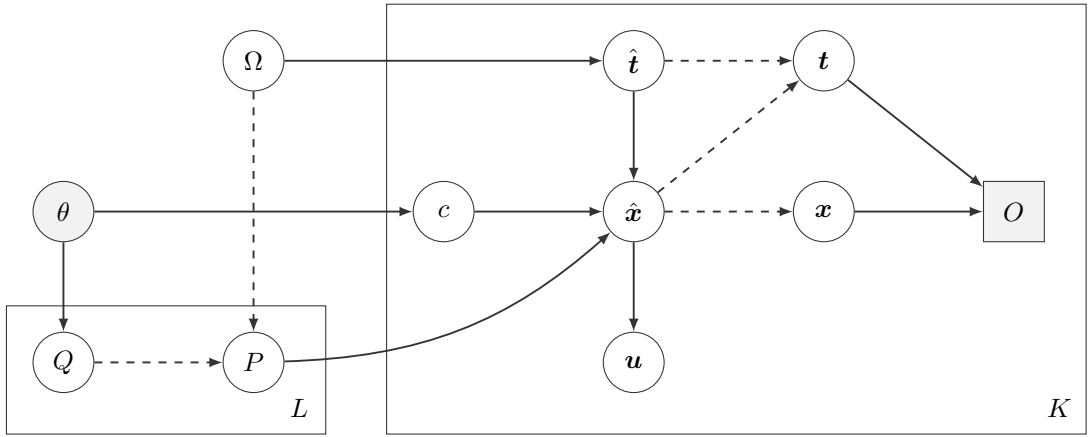
\begin{figure}[ht]
\vskip 0.1in
\begin{center}
\begin{tikzpicture}
\filldraw[white,draw=black!80!white] (4.25cm,2.75cm) rectangle ++ (9.2,-5.7);
\node[] at (13.1cm,-2.6cm) {$K$};
\filldraw[white,draw=black!80!white] (-0.75cm,-1.25cm) rectangle ++ (4.2,-1.7);
\node[] at (3.1cm,-2.6cm) {$L$};
\node (the) at (0cm,0cm) [shape=circle,fill=black!05!white,draw=black!80!white,minimum size=0.8cm] { $\theta$}; 
\node (ome) at (2.5cm,2cm) [shape=circle,draw=black!80!white,minimum size=0.8cm] { $\Omega$}; 
\node (Quni) at (2.5cm,-2cm) [shape=circle,draw=black!80!white,minimum size=0.8cm] { $P$}; 
\node (Q) at (0.cm,-2cm) [shape=circle,draw=black!80!white,minimum size=0.8cm] { $Q$}; 
\node (tvir) at (7.5cm,2cm) [shape=circle,draw=black!80!white,minimum size=0.8cm] { $\hat{\boldsymbol{t}}$}; 
\node (xvir) at (7.5cm,0cm) [shape=circle,draw=black!80!white,minimum size=0.8cm] { $\hat{\boldsymbol{x}}$}; 
\node (c) at (5cm,0cm) [shape=circle,draw=black!80!white,minimum size=0.8cm] { $c$}; 
\node (u) at (7.5cm,-2cm) [shape=circle,draw=black!80!white,minimum size=0.8cm] { $\boldsymbol{u}$}; 
\node (t) at (10cm,2cm) [shape=circle,draw=black!80!white,minimum size=0.8cm] { $\boldsymbol{t}$}; 
\node (x) at (10cm,0cm) [shape=circle,draw=black!80!white,minimum size=0.8cm] { $\boldsymbol{x}$}; 
\node (Y) at (12.5cm,0cm) [shape=rectangle,fill=black!05!white,draw=black!80!white,minimum size=0.8cm] { $O$}; 

\draw [>=latex,->,black!80!white,line width=0.25mm] (the) to (Q);
\draw [>=latex,->,black!80!white,line width=0.25mm] (the) to (c);
\draw [>=latex,->,black!80!white,line width=0.25mm,dashed] (Q) to (Quni);
\draw [>=latex,->,black!80!white,line width=0.25mm,dashed] (ome) to (Quni);
\draw [>=latex,->,black!80!white,line width=0.25mm] (ome) to (tvir);
\draw [>=latex,->,black!80!white,line width=0.25mm] (c) to (xvir);
\draw [>=latex,->,black!80!white,line width=0.25mm] (Quni) to [bend right=20] (xvir);
\draw [>=latex,->,black!80!white,line width=0.25mm] (tvir) to (xvir);
\draw [>=latex,->,black!80!white,line width=0.25mm] (xvir) to (u);
\draw [>=latex,->,black!80!white,line width=0.25mm,dashed] (tvir) to (t);
\draw [>=latex,->,black!80!white,line width=0.25mm,dashed] (xvir) to (t);
\draw [>=latex,->,black!80!white,line width=0.25mm,dashed] (xvir) to (x);
\draw [>=latex,->,black!80!white,line width=0.25mm] (x) to (Y);
\draw [>=latex,->,black!80!white,line width=0.25mm] (t) to (Y);
\end{tikzpicture}
\vskip 0.1in
\caption{Plate notation for the augmented model. Here, $\theta$ denotes a vector of fixed parameters defining unspecified priors over the generator matrices and the membership variables. Dashed arrows point to variables with deterministic dependencies.} \label{plateNotation}
\end{center}
\vskip -0.1in
\end{figure}

\begin{definition} 
Let $\boldsymbol{u}=\{u_i\}_{i=1,\dots,m}$ with $u_i\in\Sigma_{\mathcal{J}}$ be a sequence of auxiliary observations at times $0\leq \hat{t}_1,\dots,\hat{t}_{m}\leq T$. We refer to a uniformized sequence $\hat{\boldsymbol{x}}$ as `compatible' with $\boldsymbol{u}$ whenever $(\hat{x}_{i-1},\hat{x}_{i}) \in \mathcal{T}^{-1}(u_i)$, for all $i=1,\dots,m$.
\end{definition}

Conditioned on $\hat{\boldsymbol{t}}$ and a set $\boldsymbol{u}$, the subset of compatible paths in $X$ agreeing with some observations $O$ becomes finite. This property will make it possible to augment the observations with a full path $X$ in large-scale models with infinite state-support. 

\subsection{A Markov Chain Monte Carlo Algorithm}
First, for an arbitrary starting $(\boldsymbol{t},\boldsymbol{x})$ with $0=t_0<\dots<t_n<T$ and a transition probability matrix $P$, note that
\begin{align}
f_{\hat{\boldsymbol{t}}}(\hat{t}_1,\dots,\hat{t}_m|\boldsymbol{t}, \boldsymbol{x},\Omega,P,O) &\propto f_{X}(\boldsymbol{t}, \boldsymbol{x} | \hat{t}_1,\dots,\hat{t}_m, \Omega,P) \cdot f_{\hat{\boldsymbol{t}}}(\hat{t}_1,\dots,\hat{t}_m|\Omega) \nonumber \\
& \propto \Omega^{m-n} \cdot \textstyle\prod_{i=0}^n P_{x_i,x_i}^{\mathcal{V}_{i}} \label{virtualJumps}
\end{align}
whenever $\boldsymbol{t}\in\{\hat{t}_0\}\cup\{\hat{t}_1,\dots,\hat{t}_m\}$; where $\mathcal{V}_{i}$ denotes the number of elements $\hat{t}_1,\dots,\hat{t}_m$ contained in $(t_i,t_{i+1})$, with $t_{n+1}=T$, and $\sum_{i=1}^n\mathcal{V}_{i-1}=m-n$. This is independent of $O$ and may be sampled by adding virtual transitions to $\boldsymbol{t}$ using successive Poisson processes with rates $\Omega \cdot P_{x_i,x_i}$, $i\in\{0,\dots,n\}$ (cf. \citet{rao13a}).

Next, states in $\hat{\boldsymbol{x}}$ can be induced given knowledge of $\hat{\boldsymbol{t}},\boldsymbol{t},\boldsymbol{x}$, and an auxiliary sequence $\boldsymbol{u}|\hat{\boldsymbol{x}}$ sampled from \eqref{aux1}-\eqref{aux2}. An augmented path contained within a \textit{slice} (cf. \citet{neal2003}) of the full space of MJP paths $X|O$ is now attainable. A tractable procedure for the task is achieved by sampling $\hat{\boldsymbol{x}}|\hat{\boldsymbol{t}},\boldsymbol{u},\Omega,P,O$ and removing virtual entries. Let $\alpha_0(x)=\mathbb{I}(x= x_0)$, $x\in\mathcal{S}$; then, 
in view of \eqref{aux1}-\eqref{aux2}, for all $i\in\{1,\dots,m\}$  and $x'\in\mathcal{S}$ s.t. $n(\bar{\mathcal{S}}_{x',i})>0$, a forward filtering procedure reduces to predictive steps
$$\alpha_{i}(x') \propto  \textstyle\sum_{x\in\bar{\mathcal{S}}_{x',i}} \phi(x,x') \cdot P_{x,x'}   \cdot  \alpha_{i-1}(x) $$
with a penalising term for unobserved transitions
$$\phi(x,x') =  (1-q_Z)^{\mathbb{I}(\mathcal{T}(x,x')\neq\varnothing)} -p,$$
and
$$\bar{\mathcal{S}}_{x',i} = \{x\in\mathcal{S}:  |Q_{x,x'}|>0,(x,x')\in\mathcal{R}_i , \alpha_{i-1}(x)>0 \}$$
for all $x'\in\mathcal{S}$; along with updates 
$$\alpha_{i}(x') \propto \alpha_{i}(x') \cdot \textstyle\prod_{r : t_r\in[\hat{t}_i,\hat{t}_{i+1})}  f_Y(Y_r|x').$$
Above, $\mathcal{R}_i\in\Sigma_{\mathcal{S}}^2$ s.t. $\mathcal{R}_i = \mathcal{T}^{-1}(Z_d) \cap \mathcal{T}^{-1}(u_i)$ if $t_d=\hat{t}_i$ for some $d\geq 1$, and $\mathcal{R}_i = \mathcal{T}^{-1}(u_i)$ otherwise. Also, $n(\cdot)$ denotes the cardinality of a set, and we use $f_Y$ to denote the density (or mass) function for measurements in $Y$.

The above equations suggest an implementation with dynamic arrays, restricting the explorable space in the presence of clamped nodes in $\boldsymbol{u}$. Note from \eqref{aux1} that these variables are defined in order to offset the computational burden in inferential tasks with few relative observations. Finally, we sample $\hat{x}_m$ from $x\in\mathcal{S}$ in proportion to $\alpha_m(x)$ and proceed backwards; i.e. for $i\in\{m-1,\dots,1\}$ and $x\in\bar{\mathcal{S}}_{\hat{x}_{i+1},i+1}$, define
$$\beta(x) \propto \phi(x,\hat{x}_{i+1})  \cdot  P_{x,\hat{x}_{i+1}} \cdot  \alpha_i(x),$$
and sample $\hat{x}_i$ in proportion to $\beta(x)$.

\begin{proposition}
Let $Q$ be a generator matrix such that Assumption \ref{AssQ} is satisfied; also, set $\Omega > \max_{x} |Q_x|$. Then, the sampler described yields an ergodic Markov chain whose stationary distribution is the full space of processes $X|O,Q$ with countably infinite state support.
\end{proposition}
\begin{proof}
Since the generator is sparse, $\hat{\boldsymbol{x}}|\hat{\boldsymbol{t}}$ is always supported within a finite product space of $n(\hat{\boldsymbol{t}})$ subsets of $\mathcal{S}$. However, $\Omega$ strictly dominates every jump intensity, and from \eqref{virtualJumps} we note that any state in $\mathcal{S}$ is accessible by sampling the number of uniformized transitions necessary. Also, the presence of auxiliary variables leaves the target marginal distribution unaltered, so that the chain will reach the desired invariant distribution, and $\mathbb{P}(\boldsymbol{x}_{k+1},\boldsymbol{t}_{k+1}|\boldsymbol{x}_{k},\boldsymbol{t}_{k},\Omega,P,O)>0$ for any pair of paths and sampling step $k\geq 0$. 
\end{proof}

\begin{figure*}[ht]
\vskip 0.1in
\begin{center}
\resizebox{\textwidth}{!}{%
\begin{tikzpicture}


\draw[->,line width=0.10mm] (0.1cm,-0.9cm) -- ++ (7.2cm,0cm);
\draw (0.1,-0.95cm) -- ++(0cm,0.1cm);
\node at (7.55cm,-0.9cm) {\small $\boldsymbol{t}$};
\node at (-0.25cm,-0.4cm) {\footnotesize $1$};
\node at (-0.25cm,0.3cm) {\footnotesize $2$};
\node at (-0.25cm,1cm) {\footnotesize $3$};
\draw[-,draw=black!20!white,line width=0.10mm] (0.1cm,-0.4cm) -- (7.2cm,-0.4cm);
\draw[-,draw=black!20!white,line width=0.10mm] (0.1cm,0.3cm) -- (7.2cm,0.3cm); 
\draw[-,draw=black!20!white,line width=0.10mm] (0.1cm,1cm) -- (7.2cm,1cm);
\draw[-,draw=black!80!white,line width=0.45mm] (0.1cm,-0.4cm) -- (0.6cm,-0.4cm); \draw[-,draw=black!80!white,line width=0.45mm] (5.3cm,-0.4cm) -- (7.2cm,-0.4cm);
\draw[-,draw=black!80!white,line width=0.45mm] (0.6cm,0.3cm) -- (2.4cm,0.3cm); 
\draw[-,draw=black!80!white,line width=0.45mm] (2.4cm,1cm) -- (5.3cm,1cm); 
\filldraw[black!80!white] (0.1cm,-0.4cm) circle [radius=0.05cm]; 
\filldraw[black!80!white] (0.6cm,0.3cm) circle [radius=0.05cm]; 
\filldraw[black!80!white] (2.4cm,1cm) circle [radius=0.05cm]; 
\filldraw[black!80!white] (5.3cm,-0.4cm) circle [radius=0.05cm]; 
\filldraw[fill=red!50!white,draw=black!60!white] (0.3cm,-0.4cm) circle [radius=0.07cm]; 
\filldraw[fill=red!50!white,draw=black!60!white] (1.2cm,0.3cm) circle [radius=0.07cm]; 
\filldraw[fill=red!50!white,draw=black!60!white] (6.6cm,-0.4cm) circle [radius=0.07cm]; 
\filldraw[fill=black!20!white,draw=black!80!white] (0.6cm,-0.9cm) circle [radius=0.06cm]; 
\filldraw[fill=black!20!white,draw=black!80!white] (2.4cm,-0.9cm) circle [radius=0.06cm]; 
\filldraw[fill=black!20!white,draw=black!80!white] (5.3cm,-0.9cm) circle [radius=0.06cm]; 


\draw[->,line width=0.10mm] (9cm,-0.9cm) -- ++ (7.2cm,0cm);
\draw (9,-0.95cm) -- ++(0cm,0.1cm);
\node at (16.55cm,-0.9cm) {\small $\boldsymbol{t}$};
\node at (8.75cm,-0.4cm) {\footnotesize $1$};
\node at (8.75cm,0.3cm) {\footnotesize $2$};
\node at (8.75cm,1cm) {\footnotesize $3$};
\draw[-,draw=black!20!white,line width=0.10mm] (9cm,-0.4cm) -- ++ (7.1cm,0cm);
\draw[-,draw=black!20!white,line width=0.10mm] (9cm,0.3cm) -- ++ (7.1cm,0cm); 
\draw[-,draw=black!20!white,line width=0.10mm] (9cm,1cm) -- ++ (7.1cm,0cm);
\draw[-,draw=black!80!white,line width=0.45mm] (9cm,-0.4cm) -- (9.5cm,-0.4cm); \draw[-,draw=black!80!white,line width=0.45mm] (14.2cm,-0.4cm) -- (16.1cm,-0.4cm);
\draw[-,draw=black!80!white,line width=0.45mm] (9.5cm,0.3cm) -- (11.3cm,0.3cm); 
\draw[-,draw=black!80!white,line width=0.45mm] (11.3cm,1cm) -- (14.2cm,1cm); 
\filldraw[black!80!white] (9cm,-0.4cm) circle [radius=0.05cm]; 
\filldraw[black!80!white] (9.5cm,0.3cm) circle [radius=0.05cm]; 
\filldraw[black!80!white] (11.3cm,1cm) circle [radius=0.05cm]; 
\filldraw[black!80!white] (14.2cm,-0.4cm) circle [radius=0.05cm]; 
\filldraw[fill=red!50!white,draw=black!60!white] (9.2cm,-0.4cm) circle [radius=0.07cm]; 
\filldraw[fill=red!50!white,draw=black!60!white] (10.1cm,0.3cm) circle [radius=0.07cm]; 
\filldraw[fill=red!50!white,draw=black!60!white] (15.5cm,-0.4cm) circle [radius=0.07cm]; 
\filldraw[fill=black!20!white,draw=black!80!white] (9.5cm,-0.9cm) circle [radius=0.06cm]; 
\filldraw[fill=black!00!white,draw=black!80!white] (10.5cm,-0.9cm) circle [radius=0.06cm]; 
\filldraw[fill=black!20!white,draw=black!80!white] (11.3cm,-0.9cm) circle [radius=0.06cm]; 
\filldraw[fill=black!00!white,draw=black!80!white] (11.8cm,-0.9cm) circle [radius=0.06cm]; 
\filldraw[fill=black!00!white,draw=black!80!white] (13.2cm,-0.9cm) circle [radius=0.06cm]; 
\filldraw[fill=black!00!white,draw=black!80!white] (13.7cm,-0.9cm) circle [radius=0.06cm]; 
\filldraw[fill=black!20!white,draw=black!80!white] (14.2cm,-0.9cm) circle [radius=0.06cm]; 
\filldraw[fill=black!00!white,draw=black!80!white] (14.7cm,-0.9cm) circle [radius=0.06cm]; 
\node at (8.9cm,-1.4cm) {\footnotesize $\boldsymbol{u:}$};
\node at (9.5cm,-1.4cm)  {\footnotesize $\uparrow$};
\node at (10.5cm,-1.4cm) {\footnotesize $\mathcal{J}$};
\node at (11.3cm,-1.4cm)  {\footnotesize $\mathcal{J}$};
\node at (11.8cm,-1.4cm) {\footnotesize $\mathcal{J}$};
\node at (13.2cm,-1.4cm) {\footnotesize $\varnothing$};
\node at (13.7cm,-1.4cm)  {\footnotesize $\mathcal{J}$};
\node at (14.2cm,-1.4cm)  {\footnotesize $\downarrow$};
\node at (14.7cm,-1.4cm)  {\footnotesize $\varnothing$};


\draw[->,line width=0.10mm] (0.1cm,-4.5cm) -- ++ (7.2cm,0cm);
\draw (0.1,-4.55cm) -- ++(0cm,0.1cm);
\node at (7.55cm,-4.5cm) {\small $\boldsymbol{t}$};
\node at (-0.25cm,-4cm) {\footnotesize $1$};
\node at (-0.25cm,-3.3cm) {\footnotesize $2$};
\node at (-0.25cm,-2.6cm) {\footnotesize $3$};
\draw[-,draw=black!20!white,line width=0.10mm] (0.1cm,-4cm) -- ++ (7.1cm,0cm);
\draw[-,draw=black!20!white,line width=0.10mm] (0.1cm,-3.3cm) -- ++ (7.1cm,0cm); 
\draw[-,draw=black!20!white,line width=0.10mm] (0.1cm,-2.6cm) -- ++ (7.1cm,0cm);
\draw[-,draw=black!80!white,line width=0.45mm] (0.1cm,-4cm) -- (0.6cm,-4cm); \draw[-,draw=black!80!white,line width=0.45mm] (5.3cm,-4cm) -- (7.2cm,-4cm);
\draw[-,draw=black!80!white,line width=0.45mm] (0.6cm,-3.3cm) -- (1.6cm,-3.3cm); 
\draw[-,draw=black!80!white,line width=0.45mm] (4.8cm,-2.6cm) -- (5.3cm,-2.6cm); 
\filldraw[black!80!white] (0.1cm,-4cm) circle [radius=0.05cm]; 
\filldraw[black!80!white] (0.6cm,-3.3cm) circle [radius=0.05cm]; 
\filldraw[black!80!white] (4.8cm,-2.6cm) circle [radius=0.05cm]; 
\filldraw[black!80!white] (5.3cm,-4cm) circle [radius=0.05cm]; 
\filldraw[fill=red!50!white,draw=black!60!white] (0.3cm,-4cm) circle [radius=0.07cm]; 
\filldraw[fill=red!50!white,draw=black!60!white] (1.2cm,-3.3cm) circle [radius=0.07cm]; 
\filldraw[fill=red!50!white,draw=black!60!white] (6.6cm,-4cm) circle [radius=0.07cm]; 
\filldraw[fill=black!20!white,draw=black!80!white] (0.6cm,-4.5cm) circle [radius=0.06cm]; 
\filldraw[fill=black!00!white,draw=black!80!white] (1.6cm,-4.5cm) circle [radius=0.06cm]; 
\filldraw[fill=black!00!white,draw=black!80!white] (2.4cm,-4.5cm) circle [radius=0.06cm]; 
\filldraw[fill=black!00!white,draw=black!80!white] (2.9cm,-4.5cm) circle [radius=0.06cm]; 
\filldraw[fill=black!00!white,draw=black!80!white] (4.8cm,-4.5cm) circle [radius=0.06cm]; 
\filldraw[fill=black!20!white,draw=black!80!white] (5.3cm,-4.5cm) circle [radius=0.06cm]; 
\node at (0cm,-5cm) {\footnotesize $\boldsymbol{u:}$};
\node at (0.6cm,-5cm)  {\footnotesize $\uparrow$};
\node at (5.3cm,-5cm)  {\footnotesize $\downarrow$};


\draw[->,line width=0.10mm] (9cm,-4.5cm) -- ++ (7.2cm,0cm);
\draw (9,-4.55cm) -- ++(0cm,0.1cm);
\node at (16.55cm,-4.5cm) {\small $\boldsymbol{t}$};
\node at (8.75cm,-4cm) {\footnotesize $1$};
\node at (8.75cm,-3.3cm) {\footnotesize $2$};
\node at (8.75cm,-2.6cm) {\footnotesize $3$};
\draw[-,draw=black!20!white,line width=0.10mm] (9cm,-4cm) -- ++ (7.1cm,0cm);
\draw[-,draw=black!20!white,line width=0.10mm] (9cm,-3.3cm) -- ++ (7.1cm,0cm); 
\draw[-,draw=black!20!white,line width=0.10mm] (9cm,-2.6cm) -- ++ (7.1cm,0cm);
\draw[-,draw=black!80!white,line width=0.45mm] (9cm,-4cm) -- (9.5cm,-4cm); \draw[-,draw=black!80!white,line width=0.45mm] (14.2cm,-4cm) -- (16.1cm,-4cm);
\draw[-,draw=black!80!white,line width=0.45mm] (11.3cm,-3.3cm) -- (13.7cm,-3.3cm);  \draw[-,draw=black!80!white,line width=0.45mm] (9.5cm,-3.3cm) -- (10.5cm,-3.3cm); 
\draw[-,draw=black!80!white,line width=0.45mm] (10.5cm,-2.6cm) -- (11.3cm,-2.6cm);; \draw[-,draw=black!80!white,line width=0.45mm] (13.7cm,-2.6cm) -- (14.2cm,-2.6cm); 
\filldraw[black!80!white] (9cm,-4cm) circle [radius=0.05cm]; 
\filldraw[black!80!white] (9.5cm,-3.3cm) circle [radius=0.05cm]; 
\filldraw[black!80!white] (10.5cm,-2.6cm) circle [radius=0.05cm]; 
\filldraw[black!80!white] (11.3cm,-3.3cm) circle [radius=0.05cm]; 
\filldraw[black!80!white] (13.7cm,-2.6cm) circle [radius=0.05cm]; 
\filldraw[black!80!white] (14.2cm,-4cm) circle [radius=0.05cm]; 
\filldraw[fill=red!50!white,draw=black!60!white] (9.2cm,-4cm) circle [radius=0.07cm]; 
\filldraw[fill=red!50!white,draw=black!60!white] (10.1cm,-3.3cm) circle [radius=0.07cm]; 
\filldraw[fill=red!50!white,draw=black!60!white] (15.5cm,-4cm) circle [radius=0.07cm]; 
\filldraw[fill=black!20!white,draw=black!80!white] (9.5cm,-4.5cm) circle [radius=0.06cm]; 
\filldraw[fill=black!20!white,draw=black!80!white] (10.5cm,-4.5cm) circle [radius=0.06cm]; 
\filldraw[fill=black!20!white,draw=black!80!white] (11.3cm,-4.5cm) circle [radius=0.06cm]; 
\filldraw[fill=black!20!white,draw=black!80!white] (13.7cm,-4.5cm) circle [radius=0.06cm]; 
\filldraw[fill=black!20!white,draw=black!80!white] (14.2cm,-4.5cm) circle [radius=0.06cm]; 
\end{tikzpicture}
}
\vskip 0in
\caption{Diagram with a fragment of a sampling iteration. Top left, 3 measurements represented by red circles and a reference path agreeing with the evidence. Top right, data augmentation with virtual transitions and auxiliary mappings across adjacent pairs of states. Bottom left and right, emptied frame along with a projection of restrictions and new sampled path, respectively. } 
\label{samplerIter}
\end{center}
\vskip -0.1in
\end{figure*}
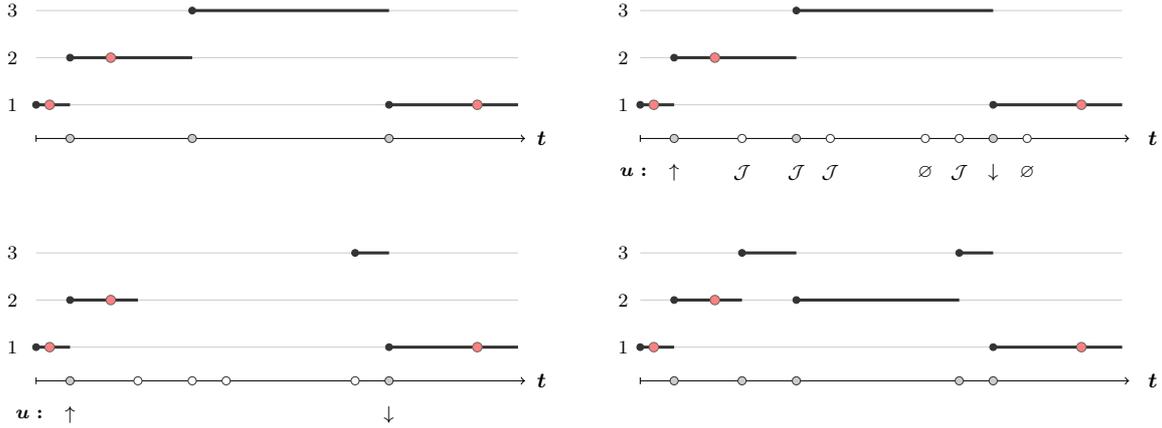

In Figure \ref{samplerIter} we observe an explanatory toy diagram with a sampling iteration, for a finite Markov chain with $3$ states, generator 
\begin{align}
Q=
\begin{bmatrix}
    -\alpha & \alpha & 0 \\
   0 & -1 & 1 \\
   1 & \delta & -(1+\delta) 
\end{bmatrix}, \label{matrixToy}
\end{align}
and fixed values for $\alpha,\delta>0$. Here, red circles represent direct process observations $Y=\{Y_r\}_{r=1,2,3}$ at the corresponding times and states. No jump data is retrieved, yet $\mathcal{J}=\{\varnothing, -1,1 \}$ is defined in order to support summary information regarding the direction (sign) of a jump, so that $\mathcal{T}(x,x')=-1+ 2 \cdot \mathbb{I}_{(x<x')}$ whenever $x'$ is accessible from $x$. On the top left, the algorithm begins with a reference pair $(\boldsymbol{t},\boldsymbol{x})$ which agrees with $Y$. Top right, we observe an augmentation to $(\hat{\boldsymbol{t}},\hat{\boldsymbol{x}},\boldsymbol{u})$, by first sampling virtual transitions with Poisson rates as seen in \eqref{virtualJumps}; later, by producing auxiliary mappings to subsets of $\mathcal{J}$ across adjacent pairs of states. There, virtual and real jump times are represented by circles over the horizontal axis, white and grey respectively. The mappings are displayed below the circles, ranging from clamped nodes with either jump sign evidence (arrow symbols) or virtual jump evidence (\textit{no-observation} symbol $\varnothing$), to open nodes (full space $\mathcal{J}$). Bottom left, we see the new frame emptied, along with a projection of the restrictions imposed by $Y,\boldsymbol{u},\hat{\boldsymbol{t}}$ on the explorable space. Bottom right, a new path is sampled using dynamic arrays, by weighting successive states over clamped epochs. 

Finally, note that the sampler will incur a considerable dependence across sequential latent paths. In practice, it is key for auxiliary variables in $\boldsymbol{u}$ to alternate across subcomponents in a model hierarchy, thus allowing reasonable variability at a marginal level and easing mixing. 

\section{Parameter inference and clustering tasks}

To explore the procedure, our experiments address various inference and clustering tasks. In this regard, we note that $Q$ is often populated by arithmetic operations involving a finite set of independent rates, and posteriors will vary across application domains. For a discussion on identifiability we refer the reader to \cite{ryden1996identifiability}; here, we restrict our attention to fully identifiable systems by means of transient studies. For later reference, we review two approaches to inference and clustering tasks.

\subsection{Gibbs sampling}

A fully Bayesian approach is completed by first sampling from the posterior distribution of process classes given some augmented paths, i.e.
\begin{align*}
\mathbb{P}(c_k=c|X^k,\boldsymbol{Q})   \propto f_{X}(X^k|Q^{c}) \cdot \pi_{c_k}(c)
\end{align*}
for $k\in\{1,\dots,K\}$. Then, it samples posterior generators given $\boldsymbol{X}=\{X^k\}_{k=1\dots,K}$ and the memberships, s.t.
\begin{align*}
f_{Q^l}(Q|\boldsymbol{X},\boldsymbol{c})  \propto  \prod_{k : c_k=l} f_{X}(X^k|Q)  \cdot \pi_{Q^l}(Q)
\end{align*}
for all $l=1,\dots,L$. For clarity, we show in Figure \ref{toyClustering} (left) results from a toy clustering task on $K=100$ processes with $L=3$ different generators. Each matrix $Q$ is of the form \eqref{matrixToy} and takes different values $\{\alpha_l,\delta_l\}_{l=1,2,3}$ as rates. In this example, given $\boldsymbol{X},\boldsymbol{c}$ in a sampler iteration, generator matrices factor across rates, s.t.
$$\alpha_l | \boldsymbol{X},\boldsymbol{c} \sim \Gamma\big( \textstyle \sum_{k : c_k=l} \psi^k_{1\rightarrow 2}, \sum_{k : c_k=l} \tau^k_1 \big)$$ and  
$$\delta_l | \boldsymbol{X},\boldsymbol{c} \sim \Gamma\big( \textstyle \sum_{k : c_k=l} \psi^k_{3\rightarrow 2}, \sum_{k : c_k=l} \tau^k_3 \big),$$
for $l=1,2,3$, assuming reasonably uninformative exponential priors. Here, $ \psi^k_{x\rightarrow x'}$ denotes the transition count between states $x$ and $x'$ in $X^k$, and $\tau^k_x$ is the time spent at state $x$. The contour plots display joint posterior densities across the different pairs of rates. Results are obtained from $4$ MCMC chains with varied starting points, $1000$ iterations and a $100$ burn-in each. 

\subsection{Centroid based procedures} \label{Clustersec}

Centroid approaches within Gibbs iterations segregate augmented paths $\boldsymbol{X}$ into $L$ groups by optimising  
$$ \underset{\boldsymbol{c}}{\arg\min} \sum_{l=1}^L\sum_{k:c_k=l} ||T(X^k)-\boldsymbol{\mu}_l||^2,$$
across membership variables $c_1,\dots,c_K$. Here, $T(\cdot)$ denotes the sufficient statistics for instantaneous rates in $Q$, and $\boldsymbol{\mu}_l$ are vectors with corresponding point estimates for each membership $1,\dots,L$. We resort to iterative refinement techniques such as \textit{k-means} and \textit{partinioning around medioids}, and alternate with data augmentation by providing cluster-level summaries to subsequent sampling steps. For recent work reviewing centroid procedures we refer the reader to \citet{kulis2011revisiting,newling2017k}.

In Figure \ref{toyClustering} (right), we display the final step in an equivalent partitioning procedure around 3-medoids. Results are obtained using the statistics $$T(X^k)= (\psi^k_{1\rightarrow 2},\psi^k_{3\rightarrow 2},\tau^k_1,\tau^k_3)$$ 
as a basis for the segregation of the underlying processes, for $k=1,\dots,100$.
\begin{figure}[h!]
\vskip 0in
\begin{center}
	\resizebox{0.8\linewidth}{5.5cm}{
    \includegraphics[width=\textwidth]{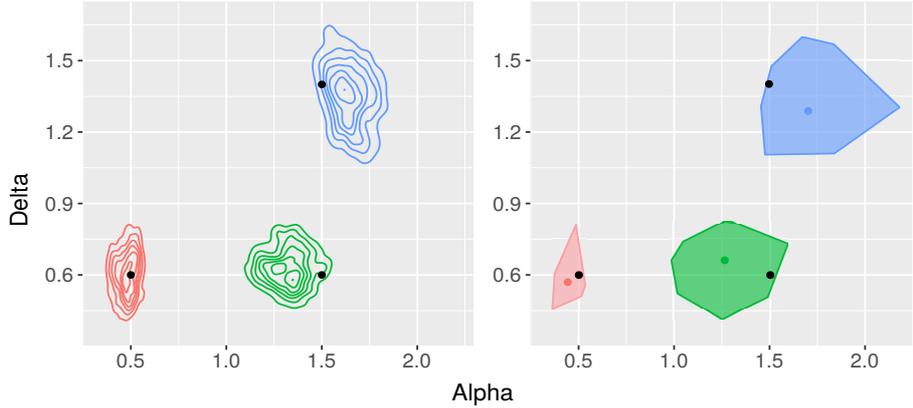} 
    }
    \vskip 0in
 \caption{Left, joint posterior densities across generator rates. Right, last step in a partitioning procedure around medoids. Dark circles correspond to real rates.} \label{toyClustering}
\end{center}
\vskip -0.2in
\end{figure}

\section{Experiments} \label{Experiments}

In the following, we discuss two experiments with real and simulated data sets. Results are produced with R and Java implementations of the sampler. In both cases, starting path and hyper-parameter values are randomized.

\subsection{Markov modulated arrival events} \label{experiment2}

We begin by addressing an augmentation task for missing patient admission times into hospital units. Observations are extracted from real-world event logs at QMC hospital, one of the busiest Accident and Emergency departments in the United Kingdom. Available data includes all $3294$ cardiology discharge event times from the 1st of January to the 31st of March, 2014. \footnote{This dataset is available on  request from http://www.hscic.gov.uk/dars.}. 

We hypothesize the existence of a Markov modulated arrival process for the admissions  (cf. \citet{scott1999bayesian,fearnhead2006exact}), with two latent \textit{regimes} that account for quiet and busy periods.  
Our aim is two-fold; first, to draw inference on patterns of high activity, second, to quantify uncertainty on intensity rates across the regimes. On a basic level, the process $X=(X_t)_{t\geq 0}=(A_t,D_t,R_t)_{t\geq 0}$ includes cumulative admission and discharge counts along with a regime indicator; so that $\mathcal{S}=\mathbb{N}_0^2\times\{1,2\}$. The transition rates in the model are given by
$$\lambda_r = Q_{a,d,r\rightarrow a+1,d,r},   \quad   \nu = Q_{a,d,r\rightarrow a,d,r'}$$
and
$$ \mu\times[L_r\wedge(a-d)] = Q_{a,d,r\rightarrow a,d+1,r},$$
for $r,r'\in\{1,2\}$, $ a,d\in\mathbb{N}_0$ and hospital manpower variables $L_r\in\mathbb{N}$, $r=1,2$. Thus, times to discharge also depend on the system regimes. The sub-component $(D_t)_{t\geq 0}$ is at all times known, by means of discharge event observations $\{Z_i\}_{i\geq 1}$ at times $t_i\geq 0, i\geq 1$. Formally, it holds $Z_i=\textstyle\mathcal{T}(\lim_{t\nearrow t_i} X_{t},X_{t_i})=D_{t_i}$ with $\mathcal{J}=\{\varnothing\}\cup\mathbb{N}$ (we use the index $i$ to avoid confusion with state variables). Note however that the pair $(A_t,R_t)_{t\geq 0}$ is always unknown and the generator $Q$ is infinite. 

The problem formulation poses a complex data augmentation problem on the space of admissions and regimes, as we must impute a large and undetermined number of variables with elaborate dependence structures. Similar tasks are often addressed in the study of infectious diseases and stochastic epidemics, and we refer the reader to \citet{SIM:SIM1912,neal2005case} for relevant literature. To proceed, we fix $\mu=0.5$ and $\nu=1/12$ to account for (i) average discharge times in the unit and (ii) expected changes between standard to \textit{out of hours} working hours. From a preliminary analysis of the data we further hypothesize that $L_1=10$ and $L_2=3$, loosely representing staff capabilities during busy and quiet periods, respectively. To ensure tractability by means of the described sampler, we define a synthetic operator $\mathcal{T}':\mathcal{S}^2\rightarrow\mathcal{J}'$ that maps state pairs to an admission event or regime switch, if any. Thus, $\mathcal{J}'=\{\varnothing\}\cup\{\nearrow_A,\leadsto_R\}\times\mathbb{N}$ supports a transition type along with its entry value. The identifiability constrain $1.25\times \lambda_1<\lambda_2$ is finally imposed, enforcing an expected minimum $25\%$ increase in admissions during busy times.

\begin{figure}[h!]
\vskip 0in
\begin{center}
	\resizebox{0.8\linewidth}{6cm}{
    \includegraphics[width=1\textwidth]{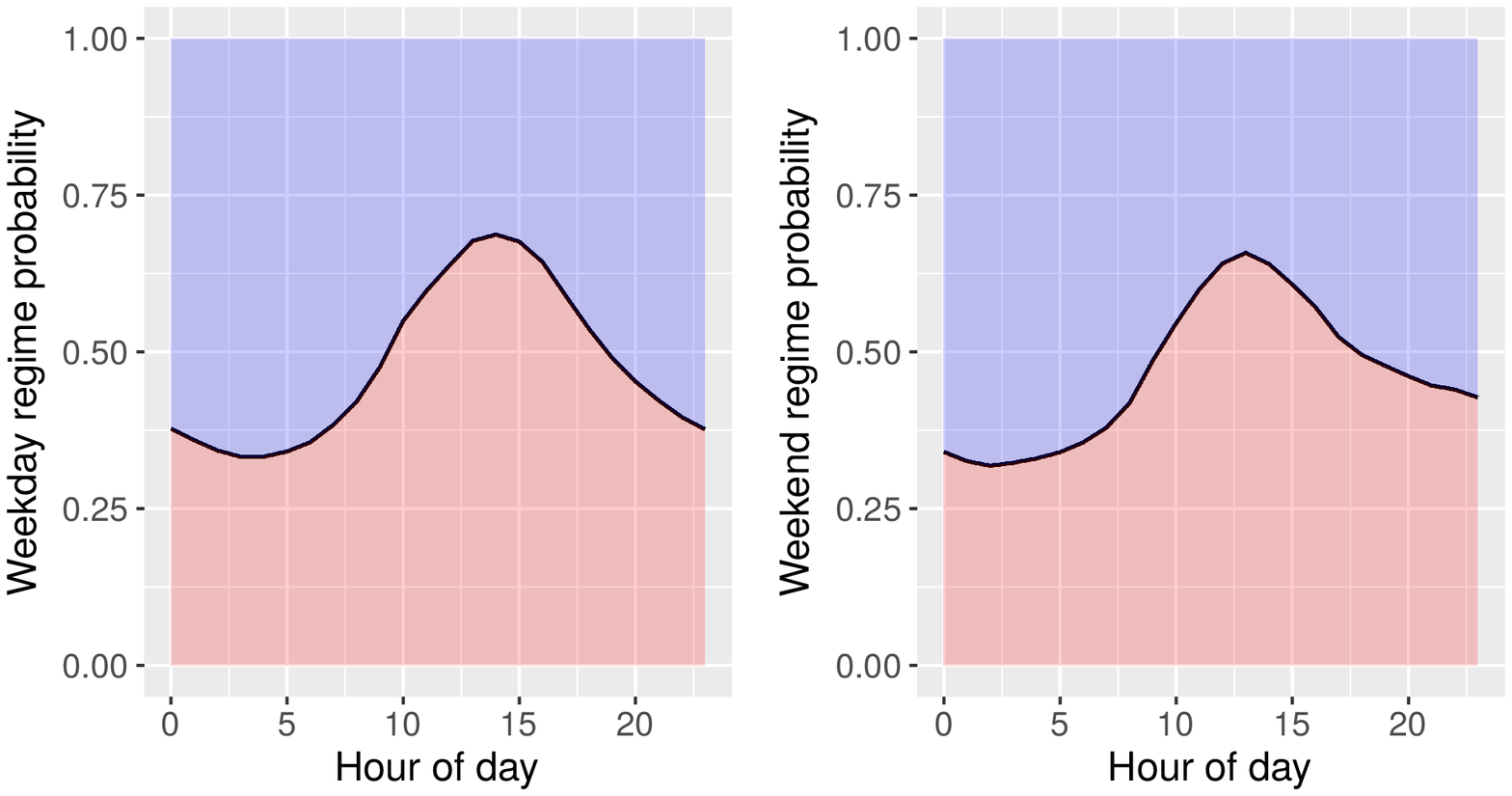}}  \\[5pt]
    \resizebox{0.81\linewidth}{5cm}{
    \includegraphics[width=1\textwidth]{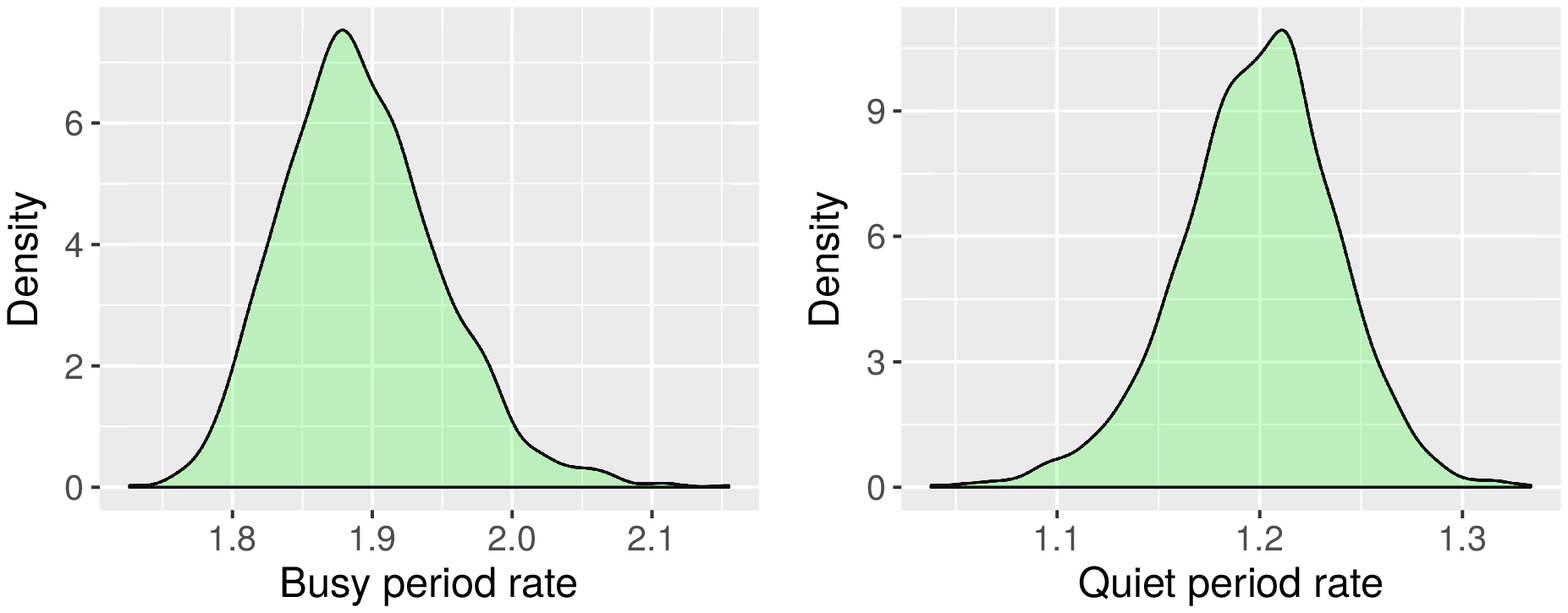} }
    \vskip 0in
   \caption{Top, probabilities for day times falling within quiet or busy working regimes. Bottom, posterior densities for arrival rates.} \label{inferenceRegimes}
\end{center}
\vskip -0.1in
\end{figure}
In Figure \ref{inferenceRegimes} we observe summaries from the various traces after $2000$ MCMC iterations, with a $500$ burn-in period, $p=0.35$ and $\Omega = 3 \cdot \max_{x} |Q_x|$. Top plots display the probabilities that any given time of the day falls within the quiet (represented in blue) or busy (red) working regime; the left hand side distribution corresponds to normal working days, the right hand side one is the weekends' equivalent. As anticipated, results suggest expected levels of admissions are higher during daylight hours, with weekends showing a slower transition into a quiet regime. Additionally, bottom plots show the posterior densities for the intensity rates that drive the unobserved arrivals. The correlation amongst traces for both rates is $0.13$, and the average increase in admissions during busy times stands at $58\%$.

To further evaluate the sampler performance, we fix regimes in $(R_t)_{t\geq 0}$, set a low scale for the dominating rate $\Omega$ and run the experiment on a reduced subset of the available data. This bounds the space of likely admissions and ensures the scalability of the baseline algorithm in \citet{rao13a}, thus allowing comparisons. In Figure \ref{ESS} (left) we display \textit{effective sample sizes} for the arrival rates on repeated tries with different values of $p$. Values are reported on 2000 MCMC iterations with a 2000 burn-in each, and $p=0$ corresponds to the baseline algorithm. There, we notice a predictable decay on effective samples as stronger dependencies are imposed between successive MJP path augmentations, which further increases the coupling with the jump rates. On the other hand, the right hand side diagram shows sample sizes adjusted for computing times. For this augmentation and inference task, we notice we can produce over twice as many effective samples in the same time span, with the appropriate tunning for the sampler.
\begin{figure}[h!]
\vskip 0in
\begin{center}
\resizebox{0.8\linewidth}{5.5cm}{
    \includegraphics[width=1\textwidth]{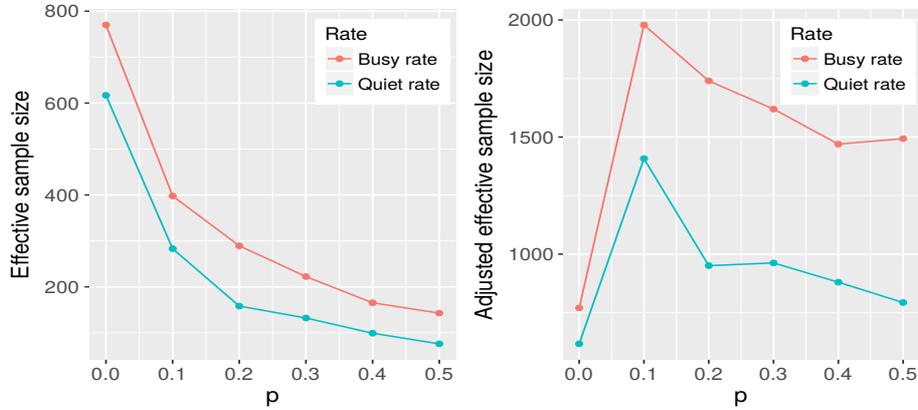}
    }
    \vskip 0in
   \caption{Left, effective sample sizes for arrival rates on repeated experiments with varying values of $p$. Right, equivalent effective sample sizes adjusted for computing times. Values reported on 2000 MCMC iterations and $p=0$ corresponds to the baseline algorithm in \citet{rao13a}.} \label{ESS}
\end{center}
\vskip -0.2in
\end{figure}

\subsection{Queueing networks}

The proposed framework is specially useful for the transient analysis of traffic flow across networks of queues, which underpin the design of modern computing systems and internet services. These networks consist of interacting components processing tasks, and their behaviour gives rise to complex stochastic systems. Uncertainty quantification tasks are challenging in common applications, and we refer the reader to \citet{sutton2011,Wang2016,Perez2017} for a review and state of the art.

The simplest instance of such network is a tandem as shown in Figure \ref{tandemNetwork}. There, shaded circles represent processing units with exponential service rates $\mu_1,\mu_2$, accompanied by \textit{queueing} areas drawn as rectangles. In this jump model, tasks arrive and wait to be processed by the different nodes in the order described by the arrow lines. In short, every task will enter station $1$, queue until the node is empty, wait for an exponentially distributed service time with rate $\mu_1$, head to station $2$, further wait until the node clears, receive an exponential service time with rate $\mu_2$ and then leave the system entirely. The underlying MJP $X$ will thus monitor the total amount of tasks waiting for service in each of the two service stations.
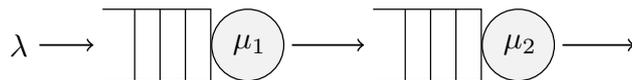
\begin{figure}[h!]
\vskip 0.1in
\begin{center}
\resizebox{0.55\linewidth}{!}{
\begin{tikzpicture}
\draw (0,0) -- ++(1.2cm,0) -- ++(0,-0.8cm) -- ++(-1.2cm,0);
\foreach \i in {1,...,3}
  \draw (1.2cm-\i*8pt,0) -- +(0,-0.8cm);
\filldraw[fill=black!05!white] (1.2cm+0.41cm,-0.4cm) circle [radius=0.4cm]; 

\draw (3cm,0cm) -- ++(1.2cm,0) -- ++(0,-0.8cm) -- ++(-1.2cm,0);
\foreach \i in {1,...,3}
  \draw (4.2cm-\i*8pt,0cm) -- +(0,-0.8cm);
\filldraw[fill=black!05!white] (4.2cm+0.41cm,0cm-0.4cm) circle [radius=0.4cm]; 

\draw[<-,line width=0.20mm] (-0.1cm,-0.4cm) -- (-0.7,-0.4cm) node[left] {$\lambda$};
\draw[->,line width=0.20mm] (2.1cm,-0.4cm) -- (2.9cm,-0.4cm);
\draw[->,line width=0.20mm] (5.1cm,-0.4cm) -- (5.9cm,-0.4cm);

\node at (1.63,-0.4cm) {$\mu_1$};
\node at (4.63,-0.4cm) {$\mu_2$};
\end{tikzpicture}}
\caption{Tandem network with $2$ servers. Shaded circles are servers accompanied by  \textit{queueing} areas pictured as rectangles. Tasks arrive in station $1$ and queue for service before heading to station $2$, where they again queue for service before completely leaving the system.} \label{tandemNetwork}
\end{center}
\vskip -0.1in
\end{figure}

In complex scenarios, there will exist multiple resources, switch units, relays or queueing areas, together processing varied tasks with different requirements; yet, a very reduced number of nodes is ever monitored. Here, we explore a reverse service diagnosis task on a single unobserved node, across a group of $200$ simulated deployments of the aforementioned tandem network; including two $M/M/1$ stations, \textit{first come first served} service disciplines and a single task class. Each network is configured with $1$ of $3$ possible entry nodes with unknown rates $\mu^c_1,c=1,2,3$. In all cases, $X=(X_t)_{t\geq 0}=(X^1_t,X^2_t)_{t\geq 0}$ includes task counts across the $2$ service stations (each including a processing unit and queueing area), and $\mathcal{S}=\mathbb{N}_0^2$. Transition rates are given by
$$\lambda = Q_{x_1,x_2\rightarrow x_1+1,x_2},\quad \mu_1^c \times \mathbb{I}_{(x_1>0)} = Q_{x_1,x_2\rightarrow x_1-1,x_2+1}$$
and
$$ \mu_2 \times \mathbb{I}_{(x_2>0)} = Q_{x_1,x_2\rightarrow x_1,x_2-1},$$
for $(x_1,x_2)\in\mathbb{N}_0^2$ and depending on the node configuration $c$. Hence, changes in marginal states may be \textit{synchronized} (a jump in a process component induces an instantaneous jump on the other). This makes inference infeasible by means of variational methods relying on assumptions of independence across model sub-components \citep{opper2008variational}.

We infer memberships and processing rates from a reduced set of transitions $Z=\{Z_i\}_{i\geq 1}$, which only include entries and departures to the network, for a limited $\% 100 \cdot q_Z$ of all tasks processed, with $q_Z=0.5$. An entry marks the arrival of a task, so that $x_1,x_2\rightarrow x_1+1,x_2$, in a departure we instead have $x_1,x_2\rightarrow x_1,x_2-1$. We denote $\mathcal{J}= \{\varnothing\}\cup\mathbb{N}\times\{0\rightarrow1,2\rightarrow0\}$, and notice that (i) the state $X$ is at all times unknown and (ii) the target node behaviour is never observed. Prior identifiability constrains are imposed to ensure $\mu_1^1<\mu_1^2<\mu_1^3$.

\begin{figure}[h!]
\vskip 0in
\begin{center}
\resizebox{0.8\linewidth}{5.5cm}{
    \includegraphics[width=1\textwidth]{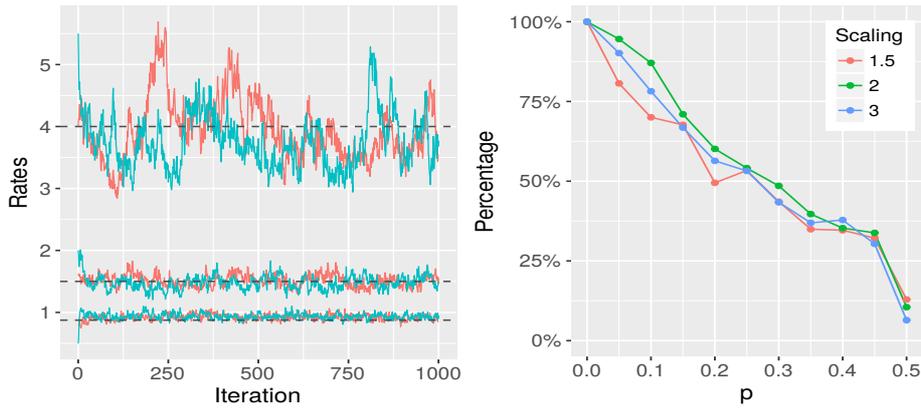}}
    \vskip 0in
   \caption{Left, MCMC trace plot for posterior service rates. Right, average percentage time decrease augmenting a path, $p=0$ serves as the reference value in each instance, corresponding to the algorithm in \citet{rao13a}.} \label{queueMCMC}
\end{center}
\vskip -0.2in
\end{figure}

In Figure \ref{queueMCMC} (left) we observe trace plots for the posterior service rates with $1000$ iterations; there, dashed horizontal lines represent the original values in the network simulations.  Results are obtained across two chains, with different tuning choices for parameters $\Omega,p$. The chains show satisfactory mixing and signs of strong serial dependencies. In the right hand side figure, we observe a summary with relative gains in computing speed during the data augmentation procedure, after repeated experiments with various dominating rate \textit{scaling levels} (so that $\Omega = K \cdot \max_{l,x} |Q^l_x|$ with $K=1.5,2,3$) and values of $p$. The diagram shows the average percentage decrease in processing time against a reference value of $p=0$ (baseline model), in each scaling instance. Again, this comparative diagram is produced with a short and restricted network realization that ensures the scalability of the baseline algorithm in \citet{rao13a}. 

\begin{table*}[t]
\setlength{\tabcolsep}{4pt}
\caption{Membership and $F_1$-score metrics in the reverse network diagnosis task, summarized across iterations in the data augmentation procedure for the different choices of clustering approaches.}
\label{tableMetrics}
\vskip 0.15in
\begin{center}
\begin{small}
\begin{sc}
\begin{tabular}{cccccccccccccc}
\toprule
 & \multirow{2}{*}{Node type} & & \multicolumn{5}{ c}{Membership summary}& & \multicolumn{5}{ c }{$F_1$ score statistics} \\[3pt]
  &&& Mean & Std &Median & Q1 & Q3 & & Mean & Std & Median& Q1 &Q3 \\ \midrule
K-means & 1 && 57.3 & 6.81 & 57 & 53  &61& & 0.71&0.07&0.70&0.66&0.76\\ 
 & 2 && 72.4 & 7.92 & 72 & 67 & 78 & & 0.58&0.06&0.58&0.54&0.62    \\ 
 &3 && 70.4 & 7.63 & 71 & 65 & 76 & & 0.55&0.04&0.54&0.52&0.58    \\ \midrule
  Pam & 1 && 54.0 & 6.98 & 54 & 50  &57& & 0.75&0.08&0.76&0.70&0.80\\ 
 & 2 && 74.8 & 8.71 & 75 & 69 & 81 & & 0.59&0.06&0.59&0.55&0.63    \\ 
 &3 && 71.2 & 8.73 & 72 & 66 & 77 & & 0.52&0.04&0.52&0.50&0.54    \\ \midrule
Gibbs & 1 && 63.7 & 6.08 & 64 & 60  & 68& & 0.64&0.04&0.65&0.61&0.67\\
  & 2 && 74.5 & 6.53 & 74 & 70 & 79 & & 0.63&0.05&0.63&0.60&0.66    \\ 
 & 3 && 61.8 & 5.72 & 62 & 58 & 66 && 0.73&0.04&0.74&0.71&0.77    \\ \bottomrule
\end{tabular}
\end{sc}
\end{small}
\end{center}
\vskip -0.1in
\end{table*}

Finally, Table \ref{tableMetrics} displays various metrics for the reverse diagnosis task. On one hand, we find a summary of network deployment memberships across the $3$ different types; after an initial burn-in period of $100$ iterations. Additionally, we include distributions of $F_1$-score metrics across the iterations in the augmentation procedure. Overall, we notice small discrepancies across clustering choices, all offering a significantly better than random partitioning. Indeed, the marginal behaviour of an individual station has a limited impact on the overall service-time distribution across an entire network (cf.  \citet{Perez2017}); and there only exist small statistical variations in the observations.

\section{Discussion}

We have presented a tractable approach for Bayesian inference with structured jump processes supported on either large or infinite state spaces. The framework is especially useful in order to address multi-component systems of coupled MJPs, which often show synchronization of events and strong dependencies across time. Our MCMC algorithm is built on uniformization principles and proceeds by exploring a multidimensional space of MJP paths, by means of restricted, alternating and sequentially correlated slices. These slices are constructed by efficiently designed measurable mappings between states and synthetic sets of jump observations. Hence, the method does not require particle filtering procedures that are often computationally intensive and find limitations due to degeneracy problems. Our experiments have shown the ability of the approach to overcome the existing computational bottleneck in various data augmentation problems, without large sacrifices in ratios of effective samples.

To date, MCMC methods are the go-to approach for posterior inference tasks with jump processes. There also exist variational methods that increase efficiency in many instances \citep{zhang2017collapsed,opper2008variational}, yet find limitations in order to quantify global system uncertainty. Hence, the need for sampling frameworks such as the one presented here is well justified.

In order to ensure mixing and to attain an optimal trade-off between computing speed and effective sample sizes, we must alternate auxiliary mappings across system subcomponents, and further tune $\Omega$ and $p$ appropriately. Increasing $p$ will prime weightings across virtual jumps and force clamped transitions on certain subsets of variables, ultimately making computationally expensive iterations less likely; yet this will increase the dependence between subsequent realizations of paths. Thus, this must be offset with increases in frequencies of virtual jumps by means of the dominating rate $\Omega$. Finally, we note that a high $p$ will hinder the algorithm from exploring the full posterior range of MJP paths. 

\bibliographystyle{apa}
\bibliography{bibliography}

\end{document}